\documentclass[11pt]{article}

\usepackage[margin=1in]{geometry}
\usepackage{lmodern}
\usepackage{microtype}
\usepackage{amsmath,amssymb,amsthm,mathtools}
\usepackage{enumitem}
\usepackage{xcolor}
\usepackage{lineno}

\usepackage{tikz}
\usepackage[ruled,vlined]{algorithm2e}
\usepackage{tcolorbox}
\tcbuselibrary{skins}

\usepackage{thmtools}
\usepackage{thm-restate}
\usepackage{hyperref}
\usepackage{cleveref}

\makeatletter
\@ifundefined{newcounteralias}{}{%
  \renewcommand\thmt@autorefsetup{%
    \@xa\def\csname\thmt@envname autorefname\@xa\endcsname
      \@xa{\thmt@thmname}}%
}
\makeatother

\colorlet{mix}{red!50!brown}

\hypersetup{colorlinks={true},linkcolor={mix},citecolor=mix,pdfborder={0 0 0}}

\setlist[itemize]{leftmargin=18pt,topsep=2pt,parsep=2pt}
\setlist[enumerate]{leftmargin=18pt,topsep=2pt,parsep=2pt}

\theoremstyle{plain}
\newtheorem{theorem}{Theorem}[section]
\newtheorem{lemma}[theorem]{Lemma}
\newtheorem{definition}[theorem]{Definition}

\newtheorem{corollary}[theorem]{Corollary}
\theoremstyle{remark}

\title{A $2$-Approximation for Directed Feedback Vertex Set in Locally Semicomplete and Quasi-Transitive Digraphs}

\author{
  Sounak Modak\\
  \small The Institute of Mathematical Sciences, HBNI, Chennai, India\\
  \small \texttt{sounakmodak@gmail.com}
  }
 \date{}

\DontPrintSemicolon
\SetAlgoSkip{smallskip}
\makeatletter
\renewcommand{\@maketitle}{%
  \newpage\null\vskip 1em
  \begin{center}%
    \let\footnote\thanks
    {\LARGE\@title\par}%
    \ifx\@author\@empty\else
      \vskip 1em
      {\large\lineskip .5em
       \begin{tabular}[t]{c}\@author\end{tabular}\par}%
    \fi
    \ifx\@date\@empty\else\vskip .5em{\large\@date\par}\fi
  \end{center}%
  \par\vskip .5em}
\makeatother

\begin{document}
\maketitle

\begin{abstract}
A \emph{directed feedback vertex set} of a digraph is a set of vertices
whose removal destroys all directed cycles. The \textsc{Directed Feedback
Vertex Set} (\textsc{DFVS}) problem asks for such a set of minimum
cardinality or minimum total weight. Although general \textsc{DFVS}
admits no constant-factor approximation under the {Unique Games Conjecture}, tournaments admit a randomized factor-$2$ approximation due to Lokshtanov et al. [SODA'20]. 
We extend this guarantee to two broader classes of
structured digraphs, both of which also contain sparse digraphs.

Our first and main result is a randomized polynomial-time
factor-$2$ approximation for weighted \textsc{DFVS} on
\emph{locally semicomplete digraphs} (\textsf{LSD}s), a class that
strictly generalizes semicomplete digraphs and tournaments. To the best
of our knowledge, this is the first non-trivial constant-factor approximation for
\textsc{DFVS} on \textsf{LSD}s, even in the unweighted setting.

Our second result is a randomized polynomial-time factor-$2$
approximation for weighted \textsc{DFVS} on \emph{quasi-transitive
digraphs}, improving the recent deterministic $9/4$-approximation
of Ghorbani and Mnich~[ICALP'26].
The algorithm follows from a simple recursive application of our
composition framework.

The
factor $2$ is optimal under the {Unique Games Conjecture}, since
tournaments are subclass of \textsf{LSD}s as well as quasi-transitive digraphs.

\nocite{DBLP:conf/icalp/GhorbaniM26}
\nocite{DBLP:journals/jacm/ChenLLOR08}
\nocite{DBLP:journals/jgt/Bang-Jensen90}
\nocite{DBLP:conf/soda/LokshtanovMMPP020}
\nocite{DBLP:conf/icalp/EsmerK25}
\end{abstract}

\clearpage

\section{Introduction}
\label{sec:intro}

A \emph{feedback vertex set} of a digraph $D$ is a set
$S\subseteq V(D)$ such that $D-S$ is acyclic. In the
\textsc{Directed Feedback Vertex Set} (\textsc{DFVS}) problem, the input
additionally contains a nonnegative vertex-weight function
$w:V(D)\to\mathbb{R}_{\ge0}$, and the objective is to find a feedback
vertex set of minimum total weight. Directed Feedback Vertex Set is one
of Karp's original $21$ NP-complete problems; its undirected counterpart
is also NP-complete
\cite{DBLP:conf/coco/Karp72,DBLP:books/fm/GareyJ79}.

Feedback vertex set problems have been studied extensively from the
perspectives of approximation
\cite{DBLP:journals/siamdm/BafnaBF99,
DBLP:journals/algorithmica/EvenNSS98},
parameterized complexity
\cite{DBLP:journals/jacm/ChenLLOR08,
DBLP:conf/focs/CyganNPPRW11,
DBLP:journals/ipl/KociumakaP14},
exact exponential-time algorithms
\cite{DBLP:journals/jacm/FominGLS19,
DBLP:conf/ictcs/Razgon07,
DBLP:journals/jco/XiaoN15},
and structural graph theory
\cite{DBLP:journals/combinatorica/ReedRST96}.

For general digraphs, the best known polynomial-time approximation ratio
for \textsc{DFVS} is
$\mathcal{O}(\log n\log\log n)$
\cite{DBLP:journals/algorithmica/EvenNSS98}, and no constant-factor
approximation exists under the Unique Games Conjecture
\cite{DBLP:journals/toc/GuruswamiL16}. This has motivated the study of
structured digraph classes on which stronger guarantees are possible.
In particular, weighted \textsc{DFVS} admits a randomized
factor-$2$ approximation on tournaments
\cite{DBLP:conf/soda/LokshtanovMMPP020}, and factor-$2$ approximations
are also known for bipartite tournaments
\cite{DBLP:journals/tcs/Zuylen11}. Further results include an
$\mathcal{O}(g)$-approximation for digraphs of genus $g$
\cite{DBLP:conf/approx/Sun24} and a $2\alpha$-approximation for
digraphs of independence number at most $\alpha$
\cite{DBLP:conf/latin/GuptaMSS24}.

\paragraph{{Locally semicomplete digraphs.}}
In this paper, our main focus is the class of
\emph{locally semicomplete digraphs} (\textsf{LSD}s), introduced by
Bang-Jensen~\cite{DBLP:journals/jgt/Bang-Jensen90}. A digraph is locally
semicomplete if the in- and out-neighborhood of every vertex induce
semicomplete digraphs. The class strictly generalizes semicomplete
digraphs and, in particular, tournaments.

Locally semicomplete digraphs retain their local density without requiring
global semicompleteness. They may contain induced directed cycles of
unbounded length, so hitting directed triangles alone does not solve
the problem. 

The structural theory of locally semicomplete digraphs was developed in
a sequence of works by Bang-Jensen and coauthors
\cite{DBLP:journals/dm/Bang-Jensen92,
DBLP:journals/dm/Bang-JensenGGV97}. A central outcome of these works is a classification theorem for LSDs~\cite{DBLP:journals/dm/Bang-JensenGGV97}. The underlying undirected graphs of locally semicomplete
digraphs are exactly the proper-circular arc graphs~\cite{DBLP:journals/jgt/Skrien82} which have many practical
applications.

A substantial literature has also investigated paths, connectivity,
hamiltonicity and decomposition problems on
locally semicomplete digraphs
\cite{DBLP:journals/jgt/Bang-JensenCM17,
DBLP:journals/jgt/Bang-JensenGV96,
DBLP:journals/jgt/Bang-JensenGY03,
DBLP:journals/jct/Bang-JensenH12,
DBLP:journals/jgt/Bang-JensenH14,
DBLP:journals/jgt/Bang-JensenM14,
DBLP:journals/jgt/Guo96,
DBLP:journals/jgt/GuoV94,
DBLP:journals/jgt/GuoV96,
DBLP:journals/ajc/GuptaGK0R10}.

\paragraph{{Quasi-transitive digraphs.}}
We also consider \emph{quasi-transitive digraphs}, in which
$xy,yz\in A(D)$ and $x\ne z$ imply that $x$ and $z$ are adjacent.
This class generalizes tournaments; its strong members admit a
canonical decomposition with a strong semicomplete quotient and
smaller quasi-transitive modules~\cite{DBLP:journals/jgt/Bang-JensenH95}.
It is also closely related to comparability graphs: an undirected graph
admits a quasi-transitive orientation if and only if it is a
comparability graph~\cite{DBLP:journals/jgt/Bang-JensenH95}.

We refer to the survey and textbook of Bang-Jensen and Gutin
\cite{DBLP:journals/jgt/Bang-JensenG98,
DBLP:books/daglib/0022205}
for a broader account of the class of locally semicomplete and quasi-transitive digraphs. In particular, both locally semicomplete and quasi-transitive digraphs can be very dense (such as tournaments or
semicomplete digraphs), or very sparse (certain orientations of paths). From an algorithmic perspective, Bang-Jensen, Maddaloni, and Saurabh
studied parameterized
complexity and kernelization of feedback set problems  on generalizations of tournaments
 including locally semicomplete and quasi-transitive digraphs~\cite{DBLP:journals/algorithmica/Bang-JensenMS16}.

\subsection*{Our Results}
\label{subsec:our-results}

We obtain the following results.

\medskip
\noindent\textbf{1. A tight $2$-approximation for locally semicomplete
digraphs.}
Our main result is a randomized polynomial-time factor-$2$
approximation for weighted \textsc{DFVS} on locally semicomplete
digraphs.

\begin{restatable}{thm}{polyapprox}
\label{thm:lsd-2approx-probhalf}
There exists a randomized algorithm $\mathcal L$ which, given a locally
semicomplete digraph $D$ on $n\ge1$ vertices and nonnegative vertex weights
$w:V(D)\to\mathbb{R}_{\ge0}$, always outputs a feedback vertex set
$X$ of $D$ and satisfies
$\Pr\!\left[
w(X)\le2\operatorname{OPT}(D,w)
\right]\ge\frac12$.
More generally, for every $\delta\in(0,1)$, the algorithm achieves
success probability at least $1-\delta$ in time
$\mathcal O(n^{18}\log(2n/\delta))$.
\end{restatable}

\medskip
\noindent\textbf{2. A $2$-approximation for quasi-transitive digraphs.}
Ghorbani and Mnich recently obtained a deterministic
$9/4$-approximation for weighted \textsc{DFVS} on quasi-transitive
digraphs~\cite{DBLP:conf/icalp/GhorbaniM26}. We show that the composition
framework developed for the round-decomposable LSD of our locally
semicomplete algorithm also applies recursively to the canonical
decomposition of quasi-transitive digraphs. This improves the
approximation factor from $9/4$ to $2$, at the cost of randomization.
More precisely, we obtain the following guarantee.

\begin{restatable}{thm}{qtapprox}
\label{thm:qt-main}
For every quasi-transitive digraph $D$ on $N\ge1$ vertices, every
nonnegative vertex-weight function
$w:V(D)\to\mathbb{R}_{\ge0}$, and every $\delta\in(0,1)$, there is a
randomized algorithm that always outputs a feedback vertex set $X$ of
$D$ and satisfies
$\Pr\!\left[
w(X)\le2\operatorname{OPT}(D,w)
\right]\ge1-\delta$.
The running time is $\mathcal O(N^{17}\log(2N/\delta))$.
\end{restatable}

Since tournaments form a subclass of both locally semicomplete and quasi-transitive digraphs, the
factor $2$ is optimal under the Unique Games Conjecture
\cite{DBLP:conf/soda/LokshtanovMMPP020}.

\paragraph{{Organization.}}
Section~\ref{sec:preliminaries} introduces the required notation and the
structural decompositions of locally semicomplete and quasi-transitive
digraphs. Section~\ref{sec:toolkit} develops the structural and approximation tools, including a composition framework for combining approximation guarantees on digraph compositions, and Section~\ref{sec:tec-overview} gives a technical
overview of the algorithms. Section~\ref{sec:lsd} presents the
factor-$2$ approximation for LSDs: Section~\ref{sec:round} handles
round-decomposable inputs, Section~\ref{sec:nonround} develops the
pivot reduction for the remaining strong inputs, and
Section~\ref{subsec:alg-corr} combines these routines and proves
Theorem~\ref{thm:lsd-2approx-probhalf}.
Section~\ref{sec:qt-2approx} proves the quasi-transitive result,
Theorem~\ref{thm:qt-main}.

\section{Preliminaries}
\label{sec:preliminaries}

We use standard graph-theoretic notation; see, e.g.,~\cite{DBLP:books/daglib/0022205} for any undefined terms.
All digraphs in this paper are finite and simple (no loops and no parallel arcs), but we allow antiparallel pairs of arcs
(i.e.\ directed $2$-cycles). For a digraph $D$, we write $V(D)$ and $A(D)$ for its vertex and arc set, respectively.
For $X\subseteq V(D)$, we denote by $D-X$ the induced subdigraph on $V(D)\setminus X$, and for $S\subseteq V(D)$
we write $D[S]$ for the subdigraph induced by $S$. The underlying undirected graph of $D$ is denoted by $U(D)$. A digraph $D$ is \emph{connected} if its underlying graph $U(D)$ is
connected, and it is \emph{strong} if every vertex is reachable from
every other vertex.

For an arc $(u,v)\in A(D)$, we call $u$ an \emph{in-neighbor} of $v$ and $v$ an \emph{out-neighbor} of $u$.
The in- and out-neighborhoods of a vertex $x$ are
$N^-_D(x):=\{\,u\in V(D)\mid (u,x)\in A(D)\,\}$,
 and $N^+_D(x):=\{\,v\in V(D)\mid (x,v)\in A(D)\,\}$,
and its in- and out-degrees are $d^-_D(x):=|N^-_D(x)|$ and $d^+_D(x):=|N^+_D(x)|$.
We will omit the subscript $D$ when the digraph is clear from context.

A vertex-weighted digraph has weights $w:V(D)\to\mathbb R_{\ge0}$;
write $w(S)=\sum_{v\in S}w(v)$.
For a nonempty finite family of vertex sets, a member is
\emph{lightest} (respectively, \emph{heaviest}) if it has minimum
(respectively, maximum) total weight among the members of the family.
Ties are broken arbitrarily. Restrictions of weights to induced
subdigraphs are implicit when unambiguous.

An \emph{independent set} of a digraph $D$ is a vertex set that induces no arcs (equivalently, it is an independent set in
$U(D)$). The \emph{independence number} of $D$ is $\alpha(U(D))$, the maximum size of an independent set in $U(D)$.
We use standard definitions of directed paths and directed cycles. A \emph{directed $3$-cycle} in a digraph $D$ is a directed cycle on three distinct vertices,
i.e.\ a set $\{a,b,c\}\subseteq V(D)$ such that (after relabeling) $(a,b),(b,c),(c,a)\in A(D)$.
We will also refer to a directed $3$-cycle as a \emph{directed triangle}, or simply a \emph{triangle}. A digraph is \emph{acyclic} if it contains no directed cycle. For a digraph $D$, a subset $C\subseteq V(D)$ is a \emph{strongly connected component} (SCC) of $D$ if
for every $u,v\in C$ there exists a directed path from $u$ to $v$ and from $v$ to $u$,
and $C$ is maximal with respect to this property.

\paragraph{{Feedback vertex set.}}
An FVS of $D$ is a set $X$ such that $D-X$ is acyclic. Write
$\operatorname{OPT}(D,w)$ for the minimum weight of such a set;
weighted \textsc{DFVS} asks for an FVS attaining this minimum.
A randomized factor-$f$ approximation always returns an FVS and
has weight at most $f\cdot\operatorname{OPT}(D,w)$ with probability at
least $1/2$.

\begin{definition}[Semicomplete and locally semicomplete~\cite{DBLP:journals/jgt/Bang-Jensen90}]
A digraph $D$ is \emph{semicomplete} if for every pair of distinct vertices $u,v$,
at least one of $uv$ or $vu$ is present.
A digraph $D$ is \emph{locally semicomplete} if for every vertex $v$,
both $D[N^+(v)]$ and $D[N^-(v)]$ are semicomplete.
\end{definition}

Note that induced subdigraphs of a locally semicomplete digraph are again locally semicomplete. 

We use the following tournament approximation, with the improved
running time from its journal version
\cite{DBLP:journals/talg/LokshtanovMMPPS21}.

\begin{theorem}[\cite{DBLP:conf/soda/LokshtanovMMPP020,DBLP:journals/talg/LokshtanovMMPPS21}]
\label{thm:tfvs-talg1}
There exists a randomized algorithm $\mathcal{T}$ which, given a tournament $D$ on $n$ vertices
and a nonnegative weight function $w:V(D)\to \mathbb{R}_{\ge 0}$, runs in time $\mathcal{O}(n^{17})$
and outputs a feedback vertex set $S$ of $D$. With probability at least $1/2$, $S$ is a $2$-approximate solution of $(D,w)$.
\end{theorem}

We use the following extension of
Theorem~\ref{thm:tfvs-talg1} to semicomplete digraphs. Apply the
standard edge-by-edge local-ratio reduction to the graph of directed
$2$-cycles, delete the zero-residual-weight vertices, and run the
tournament routine on the remaining tournament with the residual
weights. 
The amplified version is given in
Appendix~\ref{sc:amplify}.

\begin{theorem}[{\sc SEM-FVS}]
\label{thm:semfvs-2approx}
There exists a randomized algorithm $\mathcal{S}$ which, given a semicomplete digraph $H$ on $n$ vertices
and a nonnegative weight function $w:V(H)\to \mathbb{R}_{\ge 0}$, runs in time $\mathcal{O}(n^{17})$ and outputs
a feedback vertex set $X$ of $H$.
With probability at least $1/2$, the set $X$ is a $2$-approximate solution for $(H,w)$.
\end{theorem}

Following are the standard definitions~\cite{DBLP:books/daglib/0022205} that we will use throughout the paper.

\begin{definition}[Local tournament]
A digraph $D$ is a \emph{local tournament} if $D$ is oriented, and for every vertex $v\in V(D)$,
the induced subdigraphs $D[N^+(v)]$ and $D[N^-(v)]$ are tournaments. 
\end{definition}

\begin{definition}[Round local tournament]
A local tournament $R$ is \emph{round} if its vertices admit a cyclic
ordering $r_1,\ldots,r_m$ such that, for every $i$,
$N^+(r_i)=\{r_{i+1},\ldots,r_{i+d^+(r_i)}\}$
and
$N^-(r_i)=\{r_{i-d^-(r_i)},\ldots,r_{i-1}\}$,
where indices are taken modulo $m$.
Such an ordering is called a \emph{round labeling}.
\end{definition}

\begin{definition}[Composition]
\label{def:composition}
Let $R$ be a digraph with vertex set
$V(R)=\{r_1,\ldots,r_m\}$, and let
$H_1,\ldots,H_m$ be pairwise vertex-disjoint nonempty digraphs.
The \emph{composition}
$D=R[H_1,\ldots,H_m]$
has vertex set
    $V(D)=\dot\bigcup_{i=1}^m V(H_i)$
and arc set
    $A(D)=
    \bigcup_{i=1}^m A(H_i)
    \;\cup\;
    \{xy:\ x\in V(H_i),\ y\in V(H_j),\ i\ne j,
             \ r_ir_j\in A(R)\}$.
Thus, for $i\ne j$ and $x\in V(H_i),y\in V(H_j)$,
    $xy\in A(D)\quad\Longleftrightarrow\quad r_ir_j\in A(R)$.
\end{definition}

We call $R$ the \emph{quotient digraph} or simply \emph{quotient} and the subdigraphs
$H_1,\ldots,H_m$ the \emph{modules} of this composition.
Each quotient vertex $r_i$ represents the module $H_i$.
In particular, vertices belonging to the same module have identical
in- and out-neighborhoods outside that module.

\begin{definition}[Round decomposition of locally semicomplete digraphs]
\label{def:round_decomp}
A locally semicomplete digraph $D$ is \emph{round-decomposable} if
$D=R[S_1,\ldots,S_m]$
for some $m\ge 2$, where $R$ is a round local tournament and every
$S_i$ is a strong semicomplete digraph. We call $R$ the
\emph{round quotient} and the $S_i$ the \emph{modules}.
\end{definition}

We use the following structural classification of connected locally
semicomplete digraphs.

\begin{theorem}[\cite{DBLP:journals/dm/Bang-JensenGGV97}]
\label{thm:LSD-classification}
Let $D$ be a connected locally semicomplete digraph. Then exactly one of the following holds:
\begin{enumerate}[label=(\alph*)]
  \item \textbf{Round-decomposable:} $D$ is round-decomposable and admits a \emph{unique}
  round decomposition (up to cyclic permutation of the round labelling)
  $
    D = R[S_1,\dots,S_r],
  $
  where $R$ is a round local tournament on $r\ge 2$ vertices and each $S_i$ is a strong semicomplete digraph.
  \item \textbf{Non-round-decomposable:} $D$ is neither semicomplete nor round-decomposable.
  \item \textbf{Semicomplete but not round-decomposable:} $D$ is semicomplete and does not admit a
  round decomposition.
\end{enumerate}
Moreover, there is a polynomial-time algorithm that decides which case holds and produces
a corresponding certificate (in particular, a round decomposition in case (a)).
\end{theorem}

\begin{definition}[Quasi-transitive digraphs]
   A digraph $D$ is \emph{quasi-transitive} if, for all vertices
$x,y,z\in V(D)$ with $x\neq z$,
$xy,yz\in A(D)
\Longrightarrow
\{xz,zx\}\cap A(D)\neq\emptyset$.
Quasi-transitivity is hereditary under taking induced subdigraphs. 
\end{definition}

We use the following canonical decomposition of quasi-transitive digraphs.

\begin{theorem}[\cite{DBLP:journals/jgt/Bang-JensenH95}]
\label{thm:qt-canonical-decomposition}
Let $D$ be a strong quasi-transitive digraph that is not semicomplete.
Then $D$ admits a nontrivial decomposition
$D=Q[H_1,\ldots,H_t]$,
where $t\ge2$, $V(Q)=\{q_1,\ldots,q_t\}$, $Q$ is a strong
semicomplete digraph, and every $H_i$ is either a singleton or a
non-strong quasi-transitive digraph. Moreover, such a decomposition can
be found in polynomial time: its modules are the connected components
of the complement of $U(D)$, and the quotient is obtained by contraction.
\end{theorem}

\section{Structural and Approximation Toolkit}
\label{sec:toolkit}

This section collects the general structural and approximation facts
used throughout the paper, followed by the facts specific to locally
semicomplete digraphs.

\subsection{Strongly Connected Components}

The following lemmas show that cycles are contained inside SCCs and that combining approximate solutions for the SCCs preserves the approximation factor. 

\begin{lemma}\label{lem:scc}
Let $C_1,\ldots,C_k$ be the vertex sets of the strongly connected
components of a digraph $D$. Then $X\subseteq V(D)$ is a feedback
vertex set of $D$ if and only if $X\cap C_i$ is a feedback vertex set
of $D[C_i]$ for every $i\in[k]$. Also, for every nonnegative
weight function $w$, we have
\[
\operatorname{OPT}(D,w)
=
\sum_{i=1}^k
\operatorname{OPT}(D[C_i],w|_{C_i}).
\]
\end{lemma}

\begin{proof}
Every directed cycle is contained in a single strongly connected
component. The feasibility equivalence follows, and minimizing
independently over the components gives the equality.
\end{proof}

\begin{corollary}\label{cor:scc-approx}
Let $\mathcal C$ be a hereditary class of digraphs. Suppose that
$\mathcal B$ is a $\rho$-approximation for weighted {\sc DFVS} on
strong members of $\mathcal C$. Applying $\mathcal B$ to every SCC of
$D\in\mathcal C$ and taking the union always produces a feedback
vertex set of $D$. Whenever all component calls satisfy their
approximation guarantees, the resulting solution is a
$\rho$-approximation for $D$.
\end{corollary}

\begin{proof}
Let $X_i$ be the solution returned on $D[C_i]$ and put
$X:=\bigcup_{i=1}^k X_i$. Feasibility follows from
Lemma~\ref{lem:scc}. Whenever all component calls satisfy their
approximation guarantees,
\[
\begin{aligned}
w(X)=\sum_{i=1}^k w(X_i)
\le
\rho\sum_{i=1}^k\operatorname{OPT}(D[C_i],w|_{C_i})
=
\rho\,\operatorname{OPT}(D,w).
\end{aligned}
\]
\end{proof}

\subsection{Cycles and Approximation in Digraph Compositions}
\label{sec:comp}

Let $D=R[H_1,\ldots,H_m]$ be a composition and let $X\subseteq V(D)$.
For every $i\in[m]$, put
    $T_i:=V(H_i)\setminus X$,
   $I_X:=\{i\in[m]:T_i\ne\emptyset\}$.
   
\begin{lemma}
\label{lem:composition-cycles}
$D-X$ is acyclic if and only if both of the following hold:
\begin{enumerate}
    \item $H_i[T_i]$ is acyclic for every $i\in I_X$;
    \item $R[\{r_i:i\in I_X\}]$ is acyclic.
\end{enumerate}
In particular, if $X$ is a feedback vertex set of $D$, then
    $Y_X:=\{r_i:T_i=\emptyset\}$
is a feedback vertex set of $R$.
\end{lemma}

\begin{proof}
Suppose first that $D-X$ is acyclic.  Each $H_i[T_i]$ is an induced
subdigraph of $D-X$, and is therefore acyclic.  If
$R[\{r_i:i\in I_X\}]$ contained a directed cycle
    $r_{i_1}\to r_{i_2}\to\cdots\to r_{i_\ell}\to r_{i_1}$,
choose any $x_j\in T_{i_j}$ for every $j$.  The definition of composition
implies
    $x_1\to x_2\to\cdots\to x_\ell\to x_1$
in $D-X$, a contradiction.

Conversely, assume that conditions (1) and (2) hold and suppose that
$D-X$ contains a directed cycle $C$. The cycle cannot lie entirely in
one module, by condition~(1), and hence it visits at least two modules.
Traverse $C$ and replace each maximal consecutive block of vertices
belonging to the same module by the index of that module, merging the
first and last blocks if they belong to the same module. Every transition
between distinct modules corresponds to an arc of $R$. We therefore
obtain a directed closed walk in
$R[\{r_i:i\in I_X\}]$. Every directed closed walk contains a directed
cycle, contradicting condition~(2).

For the final assertion, the quotient remaining after deleting $Y_X$ is
exactly $R[\{r_i:i\in I_X\}]$, which must be acyclic by the forward
direction.
\end{proof}

Let $w:V(D)\rightarrow \mathbb{R}_{\geq 0}$ be a non-negative weight function defined on $V(D)$.  For every $i\in[m]$, let $F_i$ be a feedback vertex set of
$H_i$. 
We define
    $W_i:=w(V(H_i))$,
    $\kappa_i:=w(F_i)$ and 
    $\Delta(r_i):=W_i-\kappa_i$ where $\Delta\ge0$.  For any feedback vertex set $Y$ of $R$, we define
    $X(Y):=
    \bigcup_{r_i\in Y}V(H_i)
    \;\cup\;
    \bigcup_{r_i\notin Y}F_i$.

\begin{lemma}
\label{lem:composition-keep-delete}
The set $X(Y)$ is a feedback vertex set of $D$, and    $w(X(Y))=\sum_{i=1}^m\kappa_i+\Delta(Y)$. 
Moreover, let $\alpha,\beta\ge1$.  If
    $w(F_i)\le
    \alpha\,\operatorname{OPT}(H_i,w|_{V(H_i)})$ for every 
    $i\in[m]$
and
    $\Delta(Y)\le
    \beta\,\operatorname{OPT}(R,\Delta)$,
then
    $w(X(Y))\le
    \max\{\alpha,\beta\}\operatorname{OPT}(D,w)$.
\end{lemma}

\begin{proof}
For every module that remains nonempty after deleting $X(Y)$, the residual
module is an induced subdigraph of $H_i-F_i$ and is therefore acyclic.
The quotient induced by the nonempty residual modules is an induced
subdigraph of $R-Y$, and is therefore acyclic.  Lemma~\ref{lem:composition-cycles} now shows that $X(Y)$ is a feedback vertex
set of $D$.

If $r_i\in Y$, the construction pays
$W_i=\kappa_i+\Delta(r_i)$ on module $H_i$; otherwise it pays
$\kappa_i$.  Summing proves the displayed cost identity.

Let $X^\star$ be an optimum feedback vertex set of $(D,w)$ and define
    $Y^\star:=
    \{r_i:V(H_i)\subseteq X^\star\}$ and
    $I^\star:=
    \{i:V(H_i)\nsubseteq X^\star\}$.
By Lemma~\ref{lem:composition-cycles}, $Y^\star$ is a feedback vertex set
of $R$.  Hence
\[
\begin{aligned}
    w(X(Y))
    &=\sum_{i=1}^m\kappa_i+\Delta(Y)\\
    &\le \sum_{i=1}^m\kappa_i+
          \beta\Delta(Y^\star)\\
    &=\sum_{i\in I^\star}\kappa_i+
      \sum_{r_i\in Y^\star}
      \bigl(\beta W_i-(\beta-1)\kappa_i\bigr)\\
    &\le \sum_{i\in I^\star}\kappa_i+
          \beta\sum_{r_i\in Y^\star}W_i.
\end{aligned}
\]
For every $i\in I^\star$, the set
$X^\star\cap V(H_i)$ is a feedback vertex set of $H_i$.  Therefore
\[
    \kappa_i
    \le \alpha\,\operatorname{OPT}(H_i,w|_{V(H_i)})
    \le \alpha\,w(X^\star\cap V(H_i)).
\]
Writing $\rho:=\max\{\alpha,\beta\}$ and using that $X^\star$ contains
all of $V(H_i)$ for $r_i\in Y^\star$, we obtain
\[
\begin{aligned}
    w(X(Y))
    &\le
    \alpha\sum_{i\in I^\star}w(X^\star\cap V(H_i))
    +\beta\sum_{r_i\in Y^\star}W_i\\
    &\le \rho\,w(X^\star)
     = \rho\,\operatorname{OPT}(D,w).
\end{aligned}
\]
\end{proof}

\subsection{Local-Ratio for Vertex Cover}

The following is the standard edge-by-edge local-ratio decomposition
for weighted vertex cover, with self-looped vertices forced into the
vertex cover.

\begin{lemma}
\label{lem:vc-lr-decomposition}
Let $G=(U,E)$ be an undirected graph, possibly with self-loops, and let
$\omega:U\to\mathbb R_{\ge0}$. In $\mathcal O(|U|^3)$ time one can compute a
vertex cover $T$ and nonnegative weights
$\omega^\triangle,\omega'$ such that
$\omega=\omega^\triangle+\omega'$,
$\omega'(T)=0$ 
and for every $Y\subseteq U$ and every vertex cover $Z$ of $G$, $\omega^\triangle(Y)\le2\omega^\triangle(Z)$.
\end{lemma}



We denote the algorithm in
Lemma~\ref{lem:vc-lr-decomposition} by
$\mathcal{VC}_{\mathrm{LR}}$.

\subsection{Structural Facts for Locally Semicomplete Digraphs}

We use three consequences of the structural theory of locally
semicomplete digraphs.

\begin{lemma}[\cite{DBLP:journals/dm/Bang-JensenGGV97}]
\label{lem:nonstrong-scc-semicomplete}
Every strongly connected component of a connected, non-strong locally
semicomplete digraph induces a semicomplete digraph.
\end{lemma}

\begin{lemma}[\cite{DBLP:journals/dm/Bang-JensenGGV97}]
\label{lem:alpha2-nonround}
If $D$ is a connected, non-round-decomposable locally semicomplete
digraph, then $\alpha(U(D))\le2$.
\end{lemma}

\begin{lemma}[Vertex pancyclicity~\cite{DBLP:journals/dm/Bang-JensenGGV97}]
\label{lem:pancyclic}
Let $D$ be strong, locally semicomplete, non-round-decomposable, and
not semicomplete. Every vertex of $D$ lies on a directed cycle of
every length $\ell\in\{3,\ldots,|V(D)|\}$.
\end{lemma}

The next lemma combines the SCC structure with the independence bound.
It will allow us to handle all components outside a chosen pivot using
the semicomplete routine.

\begin{lemma}
\label{lem:nonpivot-sccs}
Let $H$ be a locally semicomplete digraph with $\alpha(U(H))\le2$,
and let $v\in V(H)$. Every SCC of $H$ not containing $v$ induces
a semicomplete digraph.
\end{lemma}

\begin{proof}
Let $K$ be the connected component of $U(H)$ containing $v$.
If $H[K]$ is strong, its only SCC contains $v$; otherwise, all its
SCCs are semicomplete by Lemma~\ref{lem:nonstrong-scc-semicomplete}.
Every other connected component $B$ of $U(H)$ induces a semicomplete
digraph: two non-adjacent vertices $x,y\in B$, together with $v$,
would form an independent set of size three. This proves the claim.
\end{proof}

\section{Technical Overview}
\label{sec:tec-overview}

We first solve each strongly connected component separately, since
every directed cycle lies in one component, thus the optimum cost of FVS gets added
(Lemma~\ref{lem:scc}).
For strong inputs, the main distinction is
whether a useful composition is already available. If a composition is available then we can handle it easily with our composition
argument for approximating FVS (\Cref{sec:comp}). Otherwise, we need structural reductions by exploiting the properties of the specific instance to obtain an approximation for FVS.

\paragraph{{Compositions and round-decomposable LSDs.}}
In a composition $D=R[H_1,\ldots,H_m]$, removing cycles inside the
modules is not enough: cycles may still pass between them through the quotient. After computing an internal FVS $F_i$ in each module, we may
either keep this solution
or delete the entire module,
 which will cost an additional
$\Delta(r_i)=w(V(H_i))-w(F_i)$.
An FVS of the quotient, weighted by $\Delta$, specifies which modules
to delete completely. Lemma~\ref{lem:composition-keep-delete} shows
that this construction combines the module and quotient approximation
factors by their maximum, rather than their product. For a round
decomposition, the modules are semicomplete and admit a factor-$2$
approximation for FVS (\Cref{thm:semfvs-2approx}) and the round quotient can be solved exactly: its lightest
in-neighborhood is an optimum FVS (Lemma~\ref{lem:round-exact}).
Together, it gives us the  factor $2$ approximation for FVS in the whole composition of  round decomposable LSDs (\Cref{lem:composition-keep-delete}).

\paragraph{{The non-round-decomposable case.}}
The main structural difficulty arises when the strong input is neither
semicomplete nor round-decomposable. We first pick a vertex $v$ as a
possible pivot. Any FVS which does not contain $v$ must contain the other endpoint of every directed $2$-cycle
through $v$, and one of the other two vertices of every triangle through $v$. These requirements of the pivot vertex $v$ form a vertex-cover
instance in the auxiliary graph $G_v$ (defined at the beginning of \Cref{sec:nonround}). Now, applying local ratio, we get a vertex cover $T$ of $G_v$
and split the weights into a charged part and a residual part under
which $T$ has weight zero. Deleting $T$ therefore removes all directed
$2$- and $3$-cycles through $v$.

This local change has a global structural consequence. In $D-T$, the
pivot's SCC must be a singleton or round-decomposable (\Cref{lem:pivot-residual}).
The fact that the independence number of the instance is bounded by $2$ (\Cref{lem:alpha2-nonround}), together with the structure of non-strong LSDs, makes
every other SCC semicomplete (Lemma~\ref{lem:pivot-residual}). All
remaining components can thus be handled directly in an iterative manner. If $v$ lies outside an optimum FVS,
the bounds for the charged and residual weights combine to give
factor $2$ (Lemma~\ref{lem:pivot-candidate}). An optimum can always
leave at least one vertex, so trying every pivot and returning the
lightest solution suffices.

\paragraph{{Quasi-transitive digraphs.}}
The canonical decomposition of strong quasi transitive digraphs provides smaller quasi-transitive
modules and a semicomplete quotient
(Theorem~\ref{thm:qt-canonical-decomposition}). We solve the modules
recursively and apply the semicomplete routine to the quotient with
the same additional deletion costs as above. The composition \Cref{lem:composition-keep-delete} preserves factor $2$ at
every level, regardless of the recursion depth. 
We repeat the semicomplete routine independently enough times to ensure
that, with high probability, every call succeeds, even when its input
weights depend on the outcomes of earlier calls.

\section{A \texorpdfstring{$2$}{2}-Approximation for DFVS in LSDs}
\label{sec:lsd}

We first handle round-decomposable instance and then use the pivot
reduction for strong instances that are neither semicomplete nor
round-decomposable. Finally, we combine these routines with
\textsc{SEM-FVS} and SCC decomposition to obtain the full algorithm.

\subsection{{The Round-Decomposable Case}}
\label{sec:round}

Let $D=R[S_1,\ldots,S_m]$ be a round decomposition.
We first $2$-approximate FVS within each semicomplete module and solve the
round quotient exactly. Thereafter, we use Lemma~\ref{lem:composition-keep-delete}
to combine these solutions into a factor-$2$ approximation for FVS in $D$.

\subsubsection{{An Exact Algorithm for the Round Quotient}}

The round ordering makes every in-neighborhood an FVS. The next lemma
shows that choosing a lightest one is optimal, even with vertex weights.

\begin{lemma}
\label{lem:round-exact}
Let $R$ be a nonempty round local tournament with nonnegative vertex
weights $\Delta$. Every in-neighborhood and every out-neighborhood
is an FVS of $R$, and
\[
\operatorname{OPT}(R,\Delta)
=\min_{v\in V(R)}\Delta(N_R^-(v))
=\min_{v\in V(R)}\Delta(N_R^+(v)).
\]
An optimal FVS can be found in $\mathcal O(|V(R)|+|A(R)|)$ time.
\end{lemma}

\begin{proof}
Fix $v$ and cut a round labeling immediately before $v$.
If an arc $x\to y$ is backward in the resulting linear order,
the forward cyclic interval from the successor of $x$ through $y$
contains $v$. Every vertex in this interval is an out-neighbor of
$x$, so $x\in N_R^-(v)$. Deleting $N_R^-(v)$ therefore removes all
backward arcs and leaves an acyclic digraph.

Choose an optimum FVS $Z$ with $R-Z$ nonempty. Such an optimum
exists because weights are nonnegative and leaving a single vertex
is feasible. Let $s$ be a source of $R-Z$. Then
$N_R^-(s)\subseteq Z$, and hence
$\Delta(N_R^-(s))\le\Delta(Z)$. Together with feasibility of every
in-neighborhood, this proves the first equality. Reversing all arcs
proves the out-neighborhood statement. Finally, all in-neighborhood
weights can be accumulated in one scan of the arcs; output a
minimum-weight one.
\end{proof}

\subsubsection{{The Round Algorithm and Its Analysis}}

The repetition count $k$ is a parameter so that the same subroutine
can be used with different error budgets. All semicomplete calls use
fresh randomness.

\begin{algorithm}[!htbp]
\caption{\textsc{RoundDecompFVS}$(D,w,k)$}
\label{alg:round}
\KwIn{A round decomposition $D=R[S_1,\ldots,S_m]$, weights $w\ge0$,
and an integer $k\ge1$}
\KwOut{A feedback vertex set of $D$}
\For{$i=1$ \KwTo $m$}{
  $F_i\gets\textsc{AmplifiedSEM-FVS}(S_i,w|_{V(S_i)},k)$\;
  $\kappa_i\gets w(F_i)$; $\Delta(r_i)\gets w(V(S_i))-\kappa_i$\;
}
Choose $r_j\in\arg\min_{r\in V(R)}\Delta(N_R^-(r))$ and set $Y\gets N_R^-(r_j)$\;
\Return{$\displaystyle\bigcup_{r_i\in Y}V(S_i)\;\cup\!\bigcup_{r_i\notin Y}F_i$}\;
\end{algorithm}

\begin{theorem}
\label{thm:round-2approx}
Algorithm~\ref{alg:round} always returns an FVS. Whenever every
module call satisfies its factor-$2$ guarantee, its output is a
factor-$2$ approximation for $(D,w)$.
\end{theorem}

\begin{proof}
Every module call returns an FVS $F_i\subseteq V(S_i)$, regardless
of whether its approximation guarantee holds. The non negativity of $w$ gives $\Delta(r_i)=w(V(S_i))-w(F_i)\ge0$.
By Lemma~\ref{lem:round-exact}, the chosen set $Y$ is an optimum FVS
of the quotient under these realized weights. The feasibility part of
Lemma~\ref{lem:composition-keep-delete} now shows that the returned
set is always an FVS of $D$. Whenever all module calls satisfy their
factor-$2$ guarantees, the same lemma applies with $\alpha=2$ and
$\beta=1$, giving the claimed approximation factor.
\end{proof}

\begin{theorem}
\label{thm:rounddecomp-2approx-probhalf}
For a given round decomposition $D=R[S_1,\ldots,S_m]$ on $n$
vertices, nonnegative weights $w$, and $\varepsilon\in(0,1)$,
Algorithm~\ref{alg:round} with
$k=\lceil\log_2(2m/\varepsilon)\rceil$ always returns an FVS $X$,
satisfies
$\Pr[w(X)\le2\operatorname{OPT}(D,w)]\ge1-\varepsilon$,
and runs in time $\mathcal O(n^{17}\log(2m/\varepsilon))$.
\end{theorem}

\begin{proof}
Feasibility holds by Theorem~\ref{thm:round-2approx}. By
Lemma~\ref{lem:tfvs-amplify}, each amplified module call fails to
satisfy its approximation guarantee with probability at most $2^{-k}$.
The union bound over all the modules gives us a probability of at most
$m2^{-k}\le\varepsilon/2\le\varepsilon$ that any module call fails.
Outside this event, Theorem~\ref{thm:round-2approx} gives the
factor-$2$ guarantee for the returned solution.

For the running time, write $n_i=|V(S_i)|$. Since the modules
partition $V(D)$, we have $\sum_i n_i=n$ and
$\sum_i n_i^{17}\le n^{17}$. The amplified calls consequently take
$\mathcal O(kn^{17})$ time in total. Computing the additional
deletion costs, solving the quotient, and constructing the final FVS
take $\mathcal O(n+m^2)$ additional time, which is dominated by this
bound. Substituting the chosen value of $k$ proves the assertion.
\end{proof}

\subsection{The Non-Round-Decomposable Case}
\label{sec:nonround}

Let $D$ be strong, locally semicomplete, neither semicomplete nor
round-decomposable, and let $v\in V(D)$.
Define the undirected graph $G_v$ on $V(D)\setminus\{v\}$ as follows:
put a loop at $x$ if $v,x$ form a directed $2$-cycle, and put an
edge $xy$ if $v,x,y$ support a directed triangle. Thus every FVS of
$D$ avoiding $v$ is a vertex cover of $G_v$.

Apply Lemma~\ref{lem:vc-lr-decomposition} to $(G_v,w)$, obtaining
$T,w^\triangle,w'$. Extend the weights by
$w^\triangle(v)=0$ and $w'(v)=w(v)$, so
\[
w=w^\triangle+w',\qquad w'(T)=0.
\]
Put $H=D-T$. Since $T$ covers $G_v$, no directed $2$- or $3$-cycle
of $H$ contains $v$.

\begin{lemma}
\label{lem:pivot-residual}
In $H$, the SCC containing $v$ is either a singleton or
round-decomposable. Every other SCC is semicomplete.
\end{lemma}

\begin{proof}
Let $C$ be the subdigraph induced by the SCC containing $v$.
There is nothing to prove about $C$ if it is a singleton, so assume
$|V(C)|>1$. Since $C$ is strong, it contains a directed cycle
through $v$. Choose a shortest such cycle,
$v=x_0\to x_1\to\cdots\to x_{\ell-1}\to v$.
The construction of $H$ excludes lengths $2$ and $3$, so $\ell\ge4$.
If $C$ were semicomplete, $v$ and $x_2$ would be adjacent. The arc
$v\to x_2$ would shorten the chosen cycle, whereas $x_2\to v$
would form the triangle $v\to x_1\to x_2\to v$. Both are
impossible, so $C$ is not semicomplete. It remains locally
semicomplete because it is induced in $D$. If it were also
non-round-decomposable, vertex pancyclicity
(Lemma~\ref{lem:pancyclic}) would give a triangle through $v$.
Therefore $C$ is round-decomposable.

For the other SCCs, Lemma~\ref{lem:alpha2-nonround} gives
$\alpha(U(D))\le2$. Vertex deletion cannot increase the independence
number, so $\alpha(U(H))\le2$ as well. Applying
Lemma~\ref{lem:nonpivot-sccs} to $H$ shows that every SCC not
containing $v$ is semicomplete.
\end{proof}

We can consequently solve all SCCs of $H$ using only the semicomplete
and round-decomposable routines. In particular, we do not recurse on  arbitrary locally semicomplete instances.

\begin{algorithm}[!htbp]
\caption{\textsc{PivotCandidate}$(D,w,v,k)$}
\label{alg:pivot-candidate}
\KwIn{A strong LSD $D$ that is neither semicomplete nor round-decomposable,
weights $w\ge0$, a pivot $v$, and an integer $k\ge1$}
\KwOut{A feedback vertex set of $D$}
Construct $G_v$ and compute $(T,w^\triangle,w')$ by local ratio as above\;
$H\gets D-T$; $F\gets\emptyset$\;
\ForEach{SCC vertex set $Q$ of $H$ with $|Q|\ge2$}{
  \eIf{$v\in Q$}{
    Compute a round decomposition of $H[Q]$\;
    $F_Q\gets\textsc{RoundDecompFVS}(H[Q],w'|_Q,k)$\;
  }{
    $F_Q\gets\textsc{AmplifiedSEM-FVS}(H[Q],w'|_Q,k)$\;
  }
  $F\gets F\cup F_Q$\;
}
\Return{$T\cup F$}\;
\end{algorithm}

\begin{lemma}
\label{lem:pivot-candidate}
Algorithm~\ref{alg:pivot-candidate} always returns an FVS $X_v$ of
$D$. If $v\notin X^\star$ for some optimum FVS $X^\star$ of
$(D,w)$, then $w(X_v)\le2w(X^\star)$ whenever all its amplified
semicomplete calls satisfy their factor-$2$ guarantees.
\end{lemma}

\begin{proof}
By Lemma~\ref{lem:pivot-residual}, we know every nontrivial SCC of $H$ is either round-decomposable or semicomplete. As each subroutine
always returns an FVS, and singleton SCCs contain no cycles, thus their outputs form an FVS $F$ of $H$ by
Lemma~\ref{lem:scc}. Hence $X_v=T\cup F$ is always feasible for $D$.

Recall the construction of $G_v$ at the very beginning of \Cref{sec:nonround}. Now fix an optimum FVS $X^\star$ avoiding $v$. Every $2$- or $3$-cycle
through $v$ must hit by $X^\star$ through another vertex other than $v$, so $X^\star$
is a vertex cover of $G_v$. Lemma~\ref{lem:vc-lr-decomposition},
extended by $w^\triangle(v)=0$, consequently gives
$w^\triangle(X_v)\le2w^\triangle(X^\star)$.
On the stated success event, the round-decomposable routine and the
semicomplete routine approximate FVS of each SCC within factor $2$.
Corollary~\ref{cor:scc-approx} therefore gives
$w'(F)\le2\operatorname{OPT}(H,w')\le2w'(X^\star)$, where the
last inequality holds because $X^\star\cap V(H)$ is an FVS of $H$.
These two bounds compare the charged and residual costs with the
same optimum $X^\star$. Since $T$ has zero residual weight, they add
to give
\[
w(X_v)=w^\triangle(X_v)+w'(F)
\le2w^\triangle(X^\star)+2w'(X^\star)=2w(X^\star).
\]
\end{proof}


\subsection{The Full Algorithm}
\label{subsec:alg-corr}

We now combine the three cases and process each SCC independently.
For a strong input that is neither semicomplete nor round-decomposable,
we try every pivot, always starting with the original weights of that
SCC, and keep the lightest solution. 

\begin{algorithm}[!htbp]
\caption{\textsc{LSD-FVS}$(D,w,\delta)$}
\label{alg:lsd-fvs}
\KwIn{An LSD $D$ on $n$ vertices, weights $w\ge0$, and $\delta\in(0,1)$}
\KwOut{A feedback vertex set of $D$}
\lIf{$n\le1$}{\Return{$\emptyset$}}
$k\gets\lceil\log_2(2n^2/\delta)\rceil$; $X\gets\emptyset$\;
\ForEach{SCC vertex set $C$ of $D$ with $|C|\ge2$}{
  $J\gets D[C]$\;
  \eIf{$J$ is semicomplete}{
    $X_C\gets\textsc{AmplifiedSEM-FVS}(J,w|_C,k)$\;
  }{
    Test round-decomposability and compute a decomposition if one exists\;
    \eIf{$J$ is round-decomposable}{
      $X_C\gets\textsc{RoundDecompFVS}(J,w|_C,k)$\;
    }{
      \ForEach{$v\in C$}{
        $X_v\gets\textsc{PivotCandidate}(J,w|_C,v,k)$ using fresh randomness\;
      }
      Choose $v^\star\in\arg\min_{v\in C}w(X_v)$ and set $X_C\gets X_{v^\star}$\;
    }
  }
  $X\gets X\cup X_C$\;
}
\Return{$X$}\;
\end{algorithm}

\subsubsection{{Correctness, Success Probability, and Running Time}}

\polyapprox*

\begin{proof}
We analyze Algorithm~\ref{alg:lsd-fvs}. First of all, every candidate is feasible
by Lemma~\ref{lem:pivot-candidate} and the guarantees of the
semicomplete and round-decomposable routines. Their union over the
SCCs is therefore an FVS of $D$ (\Cref{lem:scc}).

Suppose all amplified semicomplete calls satisfy their factor-$2$
guarantees. A strong component that is semicomplete or
round-decomposable is then approximated within factor $2$ for FVS.
For any other component $J$, choose an optimum FVS $X^\star$ that
does not contain all vertices. Such an optimum exists since weights
are nonnegative and a single vertex is acyclic. At least one enumerated
pivot $v$ therefore lies outside $X^\star$, and
Lemma~\ref{lem:pivot-candidate} gives
$w(X_v)\le2\operatorname{OPT}(J,w)$.
Thus we take the lightest solution among all the enumerated pivots. Summing over SCCs proves the
global factor-$2$ guarantee by Lemma~\ref{lem:scc}.

Now we bound the number of amplified semicomplete calls. Within a pivot
candidate on a component of size $a$, the inputs to these calls are
pairwise disjoint: they are the nonpivot SCCs and the semicomplete
modules of the SCC containing the pivot. Hence there are at most $a$ such calls.
There are $a$ pivots, so this component contributes at most $a^2$
calls. The semicomplete and round-decomposable subroutines also use at
most $a^2$ calls. If the original SCC sizes are $a_1,\ldots,a_s$,
the total is at most $\sum_j a_j^2\le n^2$.
Each call uses fresh randomness and, conditional on its input, fails
with probability at most $2^{-k}$. Therefore
\[
\Pr[\text{some amplified call fails}]
\le n^2 2^{-k}\le\delta/2\le\delta.
\]
Thus we have the claimed success probability.

For the running time, let $b_1,\ldots,b_q$ be the semicomplete input
sizes within one pivot candidate on $a$ vertices. Since
$\sum_i b_i\le a$, these calls take
$\mathcal O\!\left(k\sum_i b_i^{17}\right)
\subseteq\mathcal O(ka^{17})$
time. Constructing $G_v$, performing local ratio, and computing SCCs
take $\mathcal O(a^3)$ time. The round-decomposition algorithm
of~\cite[Proposition~3.9]{DBLP:journals/dm/Bang-JensenGGV97}
admits an $\mathcal O(a^4)$ implementation.
The exact quotient solver adds
at most $\mathcal O(a^2)$ time. All these operations are dominated
by the semicomplete calls, so all $a$ candidates take
$\mathcal O(ka^{18})$ time. 
Summing over original SCCs and using
$\sum_j a_j^{18}\le n^{18}$ gives
$\mathcal O(kn^{18})=\mathcal O(n^{18}\log(2n/\delta))$.

Setting $\delta=1/2$ gives success probability at least $1/2$ in
$\mathcal O(n^{18}\log(2n))$ time.
\end{proof}

\section{A Randomized 2-Approximation for Quasi-Transitive Digraphs}
\label{sec:qt-2approx}

We apply Lemma~\ref{lem:composition-keep-delete} recursively to the
canonical decomposition of Theorem~\ref{thm:qt-canonical-decomposition}.
For the original input size $N\ge1$ and $\delta\in(0,1)$, let
$\mathcal S^{\mathrm{amp}}(H,w;N,\delta)$ run $\mathcal S(H,w)$
independently $\lceil\log_2(4N/\delta)\rceil$ times and return the
lightest output. By Lemma~\ref{lem:tfvs-amplify}, it is always feasible,
has failure probability at most $\delta/(4N)$, and takes
$\mathcal O(|V(H)|^{17}\log(2N/\delta))$ time.

\subsection{The Algorithm}

We now give the recursive algorithm. The parameter $N$ denotes the number
of vertices in the original input and is passed unchanged through all
recursive calls. It is used only for amplification.

\begin{algorithm}[!htbp]
\caption{\textsc{QT-FVS}$(D,w,N,\delta)$}
\label{alg:qt-fvs}
\KwIn{A quasi-transitive digraph $D$, nonnegative weights $w$,
original input size $N$, and $\delta\in(0,1)$}
\KwOut{A feedback vertex set of $D$}
\lIf{$|V(D)|\le1$}{\Return{$\emptyset$}}
Compute the SCC vertex sets $C_1,\ldots,C_s$ of $D$\;
\If{$s\ge2$}{
  \For{$j=1$ \KwTo $s$}{
    $X_j\gets\textsc{QT-FVS}(D[C_j],w|_{C_j},N,\delta)$\;
  }
  \Return{$\bigcup_{j=1}^sX_j$}\;
}
\lIf{$D$ is semicomplete}{\Return{$\mathcal S^{\mathrm{amp}}(D,w;N,\delta)$}}
Compute $D=Q[H_1,\ldots,H_t]$, with $V(Q)=\{q_1,\ldots,q_t\}$,
as in Theorem~\ref{thm:qt-canonical-decomposition}\;
\For{$i=1$ \KwTo $t$}{
  $F_i\gets\textsc{QT-FVS}(H_i,w|_{V(H_i)},N,\delta)$\;
  $W_i\gets w(V(H_i))$; $\kappa_i\gets w(F_i)$; $\Delta(q_i)\gets W_i-\kappa_i$\;
}
$Y\gets\mathcal S^{\mathrm{amp}}(Q,\Delta;N,\delta)$\;
\Return{$\displaystyle\bigcup_{q_i\in Y}V(H_i)\;\cup\!\bigcup_{q_i\notin Y}F_i$}\;
\end{algorithm}

\subsection{Correctness, Approximation Guarantee, and Running Time}

\qtapprox*

\begin{proof}
We analyze Algorithm~\ref{alg:qt-fvs}, keeping $N$ fixed throughout.
We prove inductively that each recursive call is always
feasible and is a factor-$2$ approximation whenever all amplified
semicomplete calls in its subtree satisfy their guarantees.
The case $|V(D)|\le1$ is immediate. For a non-strong input, both
assertions follow from the induction hypothesis and Lemma~\ref{lem:scc}.
For a strong semicomplete input, they are precisely the guarantees of
$\mathcal S^{\mathrm{amp}}$.

Otherwise, Theorem~\ref{thm:qt-canonical-decomposition} gives
$D=Q[H_1,\ldots,H_t]$ with a semicomplete quotient and proper,
nonempty quasi-transitive modules. Induction supplies feasible $F_i$,
so $\Delta(q_i)=w(V(H_i))-w(F_i)\ge0$.
Lemma~\ref{lem:composition-keep-delete} gives feasibility
unconditionally. When the module subtrees and quotient call succeed,
the same lemma with $\alpha=\beta=2$ gives
$w(X)\le2\operatorname{OPT}(D,w)$. This proves the invariant.

Every internal node of the recursion tree has at least two children
whose nonempty vertex sets partition its input. There are at most
$N$ leaves and hence at most $2N-1$ nodes. Each node makes at most
one amplified semicomplete call, either on its input or on its quotient.
Each such call uses fresh randomness. Conditional on all preceding
choices, its input weights are fixed and its failure probability is
at most $\delta/(4N)$, even when those weights depend on the outcomes of earlier calls. The same bound therefore holds unconditionally, and
the union bound gives
\[
\Pr[\text{some semicomplete call fails}]
\le(2N-1)\frac{\delta}{4N}<\delta.
\]
The invariant proves the claimed approximation probability.

For the running time, the semicomplete inputs at leaves have total
size at most $N$. At each node using a canonical decomposition, the
quotient size equals its number of children. The sum of child counts
over all internal nodes is the number of tree edges, at most $2N-2$.
Thus, if $b_1,\ldots,b_q$ are the sizes of all semicomplete inputs,
$\sum_i b_i\le3N-2$ and $b_i\le N$. Their total running time is
\[
\mathcal O\!\left(\log(2N/\delta)\sum_i b_i^{17}\right)
\subseteq\mathcal O(N^{17}\log(2N/\delta)).
\]
SCC computation and construction of the canonical decomposition time is dominated by above bound.

Taking $\delta=1/2$ gives a randomized polynomial-time factor-$2$
approximation with success probability at least $1/2$.
\end{proof}

\section*{Acknowledgements}

The author thanks Prof.~Saket Saurabh and Sobyasachi Chatterjee
for listening to the presentation of the algorithm and identifying
errors in a preliminary version of this work. The author is
particularly grateful to Prof.~Saket Saurabh for technical insights
that substantially simplified the arguments and shortened the
original 50-page draft. These insights also helped reduce the
polynomial exponent in the running time for locally semicomplete
digraphs from $500$ to $18$.

\section*{Declaration of AI Assistance}

ChatGPT (5.5 and 5.6) was used to assist in refining the language and restructuring the exposition.

\bibliography{ref}

@inproceedings{DBLP:conf/soda/LokshtanovMMPP020,
  author       = {Daniel Lokshtanov and
                  Pranabendu Misra and
                  Joydeep Mukherjee and
                  Fahad Panolan and
                  Geevarghese Philip and
                  Saket Saurabh},
  editor       = {Shuchi Chawla},
  title        = {2-Approximating Feedback Vertex Set in Tournaments},
  booktitle    = {Proceedings of the 2020 {ACM-SIAM} Symposium on Discrete Algorithms,
                  {SODA} 2020, Salt Lake City, UT, USA, January 5-8, 2020},
  pages        = {1010--1018},
  publisher    = {{SIAM}},
  year         = {2020},
  url          = {https://doi.org/10.1137/1.9781611975994.61},
  doi          = {10.1137/1.9781611975994.61},
  bibsource    = {dblp computer science bibliography, https://dblp.org}
}

@article{DBLP:journals/algorithmica/Bang-JensenMS16,
  author       = {J{\o}rgen Bang{-}Jensen and
                  Alessandro Maddaloni and
                  Saket Saurabh},
  title        = {Algorithms and Kernels for Feedback Set Problems in Generalizations
                  of Tournaments},
  journal      = {Algorithmica},
  volume       = {76},
  number       = {2},
  pages        = {320--343},
  year         = {2016},
  url          = {https://doi.org/10.1007/s00453-015-0038-2},
  doi          = {10.1007/S00453-015-0038-2},
  bibsource    = {dblp computer science bibliography, https://dblp.org}
}

@article{DBLP:journals/dm/Bang-JensenGGV97,
  author       = {J{\o}rgen Bang{-}Jensen and
                  Yubao Guo and
                  Gregory Z. Gutin and
                  Lutz Volkmann},
  title        = {A classification of locally semicomplete digraphs},
  journal      = {Discret. Math.},
  volume       = {167-168},
  pages        = {101--114},
  year         = {1997},
  url          = {https://doi.org/10.1016/S0012-365X(96)00219-1},
  doi          = {10.1016/S0012-365X(96)00219-1},
  bibsource    = {dblp computer science bibliography, https://dblp.org}
}

@article{DBLP:journals/jgt/Bang-JensenG98,
  author       = {J{\o}rgen Bang{-}Jensen and
                  Gregory Z. Gutin},
  title        = {Generalizations of tournaments: {A} survey},
  journal      = {J. Graph Theory},
  volume       = {28},
  number       = {4},
  pages        = {171--202},
  year         = {1998},
  url          = {https://doi.org/10.1002/(SICI)1097-0118(199808)28:4\<171::AID-JGT1\>3.0.CO;2-G},
  doi          = {10.1002/(SICI)1097-0118(199808)28:4\<171::AID-JGT1\>3.0.CO;2-G},
  bibsource    = {dblp computer science bibliography, https://dblp.org}
}

@article{DBLP:journals/jgt/Bang-Jensen90,
  author       = {J{\o}rgen Bang{-}Jensen},
  title        = {Locally semicomplete digraphs: {A} generalization of tournaments},
  journal      = {J. Graph Theory},
  volume       = {14},
  number       = {3},
  pages        = {371--390},
  year         = {1990},
  url          = {https://doi.org/10.1002/jgt.3190140310},
  doi          = {10.1002/JGT.3190140310},
  bibsource    = {dblp computer science bibliography, https://dblp.org}
}

@article{DBLP:journals/dm/Bang-Jensen92,
  author       = {J{\o}rgen Bang{-}Jensen},
  title        = {On the structure of locally semicomplete digraphs},
  journal      = {Discret. Math.},
  volume       = {100},
  number       = {1-3},
  pages        = {243--265},
  year         = {1992},
  url          = {https://doi.org/10.1016/0012-365X(92)90645-V},
  doi          = {10.1016/0012-365X(92)90645-V},
  bibsource    = {dblp computer science bibliography, https://dblp.org}
}

@inproceedings{DBLP:conf/latin/GuptaMSS24,
  author       = {Sushmita Gupta and
                  Sounak Modak and
                  Saket Saurabh and
                  Sanjay Seetharaman},
  editor       = {Jos{\'{e}} A. Soto and
                  Andreas Wiese},
  title        = {Quick-Sort Style Approximation Algorithms for Generalizations of Feedback
                  Vertex Set in Tournaments},
  booktitle    = {{LATIN} 2024: Theoretical Informatics - 16th Latin American Symposium,
                  Puerto Varas, Chile, March 18-22, 2024, Proceedings, Part {I}},
  series       = {Lecture Notes in Computer Science},
  volume       = {14578},
  pages        = {225--240},
  publisher    = {Springer},
  year         = {2024},
  url          = {https://doi.org/10.1007/978-3-031-55598-5\_15},
  doi          = {10.1007/978-3-031-55598-5\_15},
  bibsource    = {dblp computer science bibliography, https://dblp.org}
}

@book{DBLP:books/daglib/0022205,
  author       = {J{\o}rgen Bang{-}Jensen and
                  Gregory Z. Gutin},
  title        = {Digraphs - Theory, Algorithms and Applications, Second Edition},
  series       = {Springer Monographs in Mathematics},
  publisher    = {Springer},
  year         = {2009},
  isbn         = {978-1-84800-997-4},
  bibsource    = {dblp computer science bibliography, https://dblp.org}
}

@inproceedings{DBLP:conf/icalp/EsmerK25,
  author       = {Baris Can Esmer and
                  Ariel Kulik},
  editor       = {Keren Censor{-}Hillel and
                  Fabrizio Grandoni and
                  Jo{\"{e}}l Ouaknine and
                  Gabriele Puppis},
  title        = {Sampling with a Black Box: Faster Parameterized Approximation Algorithms
                  for Vertex Deletion Problems},
  booktitle    = {52nd International Colloquium on Automata, Languages, and Programming,
                  {ICALP} 2025, Aarhus, Denmark, July 8-11, 2025},
  series       = {LIPIcs},
  volume       = {334},
  pages        = {39:1--39:20},
  publisher    = {Schloss Dagstuhl - Leibniz-Zentrum f{\"{u}}r Informatik},
  year         = {2025},
  url          = {https://doi.org/10.4230/LIPIcs.ICALP.2025.39},
  doi          = {10.4230/LIPICS.ICALP.2025.39},
  bibsource    = {dblp computer science bibliography, https://dblp.org}
}

@article{DBLP:journals/jacm/ChenLLOR08,
  author       = {Jianer Chen and
                  Yang Liu and
                  Songjian Lu and
                  Barry O'Sullivan and
                  Igor Razgon},
  title        = {A fixed-parameter algorithm for the directed feedback vertex set problem},
  journal      = {J. {ACM}},
  volume       = {55},
  number       = {5},
  pages        = {21:1--21:19},
  year         = {2008},
  url          = {https://doi.org/10.1145/1411509.1411511},
  doi          = {10.1145/1411509.1411511},
  bibsource    = {dblp computer science bibliography, https://dblp.org}
}

@book{DBLP:books/fm/GareyJ79,
  author       = {M. R. Garey and
                  David S. Johnson},
  title        = {Computers and Intractability: {A} Guide to the Theory of NP-Completeness},
  publisher    = {W. H. Freeman},
  year         = {1979},
  isbn         = {0-7167-1044-7},
  bibsource    = {dblp computer science bibliography, https://dblp.org}
}

@article{DBLP:journals/siamdm/BafnaBF99,
  author       = {Vineet Bafna and
                  Piotr Berman and
                  Toshihiro Fujito},
  title        = {A 2-Approximation Algorithm for the Undirected Feedback Vertex Set
                  Problem},
  journal      = {{SIAM} J. Discret. Math.},
  volume       = {12},
  number       = {3},
  pages        = {289--297},
  year         = {1999},
  url          = {https://doi.org/10.1137/S0895480196305124},
  doi          = {10.1137/S0895480196305124},
  bibsource    = {dblp computer science bibliography, https://dblp.org}
}

@article{DBLP:journals/algorithmica/EvenNSS98,
  author       = {Guy Even and
                  Joseph Naor and
                  Baruch Schieber and
                  Madhu Sudan},
  title        = {Approximating Minimum Feedback Sets and Multicuts in Directed Graphs},
  journal      = {Algorithmica},
  volume       = {20},
  number       = {2},
  pages        = {151--174},
  year         = {1998},
  url          = {https://doi.org/10.1007/PL00009191},
  doi          = {10.1007/PL00009191},
  bibsource    = {dblp computer science bibliography, https://dblp.org}
}

@inproceedings{DBLP:conf/focs/CyganNPPRW11,
  author       = {Marek Cygan and
                  Jesper Nederlof and
                  Marcin Pilipczuk and
                  Michal Pilipczuk and
                  Johan M. M. van Rooij and
                  Jakub Onufry Wojtaszczyk},
  editor       = {Rafail Ostrovsky},
  title        = {Solving Connectivity Problems Parameterized by Treewidth in Single
                  Exponential Time},
  booktitle    = {{IEEE} 52nd Annual Symposium on Foundations of Computer Science, {FOCS}
                  2011, Palm Springs, CA, USA, October 22-25, 2011},
  pages        = {150--159},
  publisher    = {{IEEE} Computer Society},
  year         = {2011},
  url          = {https://doi.org/10.1109/FOCS.2011.23},
  doi          = {10.1109/FOCS.2011.23},
  bibsource    = {dblp computer science bibliography, https://dblp.org}
}

@article{DBLP:journals/ipl/KociumakaP14,
  author       = {Tomasz Kociumaka and
                  Marcin Pilipczuk},
  title        = {Faster deterministic Feedback Vertex Set},
  journal      = {Inf. Process. Lett.},
  volume       = {114},
  number       = {10},
  pages        = {556--560},
  year         = {2014},
  url          = {https://doi.org/10.1016/j.ipl.2014.05.001},
  doi          = {10.1016/J.IPL.2014.05.001},
  bibsource    = {dblp computer science bibliography, https://dblp.org}
}

@inproceedings{DBLP:conf/ictcs/Razgon07,
  author       = {Igor Razgon},
  editor       = {Giuseppe F. Italiano and
                  Eugenio Moggi and
                  Luigi Laura},
  title        = {Computing Minimum Directed Feedback Vertex Set in O(1.9977\({}^{\mbox{n}}\))},
  booktitle    = {Theoretical Computer Science, 10th Italian Conference, {ICTCS} 2007,
                  Rome, Italy, October 3-5, 2007, Proceedings},
  pages        = {70--81},
  publisher    = {World Scientific},
  year         = {2007},
  bibsource    = {dblp computer science bibliography, https://dblp.org}
}

@article{DBLP:journals/jco/XiaoN15,
  author       = {Mingyu Xiao and
                  Hiroshi Nagamochi},
  title        = {An improved exact algorithm for undirected feedback vertex set},
  journal      = {J. Comb. Optim.},
  volume       = {30},
  number       = {2},
  pages        = {214--241},
  year         = {2015},
  url          = {https://doi.org/10.1007/s10878-014-9737-x},
  doi          = {10.1007/S10878-014-9737-X},
  bibsource    = {dblp computer science bibliography, https://dblp.org}
}

@article{DBLP:journals/combinatorica/ReedRST96,
  author       = {Bruce A. Reed and
                  Neil Robertson and
                  Paul D. Seymour and
                  Robin Thomas},
  title        = {Packing Directed Circuits},
  journal      = {Comb.},
  volume       = {16},
  number       = {4},
  pages        = {535--554},
  year         = {1996},
  url          = {https://doi.org/10.1007/BF01271272},
  doi          = {10.1007/BF01271272},
  bibsource    = {dblp computer science bibliography, https://dblp.org}
}

@inproceedings{DBLP:conf/coco/Karp72,
  author       = {Richard M. Karp},
  editor       = {Raymond E. Miller and
                  James W. Thatcher},
  title        = {Reducibility Among Combinatorial Problems},
  booktitle    = {Proceedings of a symposium on the Complexity of Computer Computations,
                  held March 20-22, 1972, at the {IBM} Thomas J. Watson Research Center,
                  Yorktown Heights, New York, {USA}},
  series       = {The {IBM} Research Symposia Series},
  pages        = {85--103},
  publisher    = {Plenum Press, New York},
  year         = {1972},
  url          = {https://doi.org/10.1007/978-1-4684-2001-2\_9},
  doi          = {10.1007/978-1-4684-2001-2\_9},
  bibsource    = {dblp computer science bibliography, https://dblp.org}
}

@article{DBLP:journals/jacm/FominGLS19,
  author       = {Fedor V. Fomin and
                  Serge Gaspers and
                  Daniel Lokshtanov and
                  Saket Saurabh},
  title        = {Exact Algorithms via Monotone Local Search},
  journal      = {J. {ACM}},
  volume       = {66},
  number       = {2},
  pages        = {8:1--8:23},
  year         = {2019},
  url          = {https://doi.org/10.1145/3284176},
  doi          = {10.1145/3284176},
  bibsource    = {dblp computer science bibliography, https://dblp.org}
}

@article{DBLP:journals/toc/GuruswamiL16,
  author       = {Venkatesan Guruswami and
                  Euiwoong Lee},
  title        = {Simple Proof of Hardness of Feedback Vertex Set},
  journal      = {Theory Comput.},
  volume       = {12},
  number       = {1},
  pages        = {1--11},
  year         = {2016},
  url          = {https://doi.org/10.4086/toc.2016.v012a006},
  doi          = {10.4086/TOC.2016.V012A006},
  bibsource    = {dblp computer science bibliography, https://dblp.org}
}

@article{DBLP:journals/tcs/Zuylen11,
  author       = {Anke van Zuylen},
  title        = {Linear programming based approximation algorithms for feedback set
                  problems in bipartite tournaments},
  journal      = {Theor. Comput. Sci.},
  volume       = {412},
  number       = {23},
  pages        = {2556--2561},
  year         = {2011},
  url          = {https://doi.org/10.1016/j.tcs.2010.10.047},
  doi          = {10.1016/J.TCS.2010.10.047},
  bibsource    = {dblp computer science bibliography, https://dblp.org}
}

@inproceedings{DBLP:conf/approx/Sun24,
  author       = {Hao Sun},
  editor       = {Amit Kumar and
                  Noga Ron{-}Zewi},
  title        = {A Constant Factor Approximation for Directed Feedback Vertex Set in
                  Graphs of Bounded Genus},
  booktitle    = {Approximation, Randomization, and Combinatorial Optimization. Algorithms
                  and Techniques, {APPROX/RANDOM} 2024, London School of Economics,
                  London, UK, August 28-30, 2024},
  series       = {LIPIcs},
  volume       = {317},
  pages        = {18:1--18:20},
  publisher    = {Schloss Dagstuhl - Leibniz-Zentrum f{\"{u}}r Informatik},
  year         = {2024},
  url          = {https://doi.org/10.4230/LIPIcs.APPROX/RANDOM.2024.18},
  doi          = {10.4230/LIPICS.APPROX/RANDOM.2024.18},
  bibsource    = {dblp computer science bibliography, https://dblp.org}
}

@article{DBLP:journals/ajc/GuptaGK0R10,
  author       = {Arvind Gupta and
                  Gregory Z. Gutin and
                  Mehdi Karimi and
                  Eun Jung Kim and
                  Arash Rafiey},
  title        = {Minimum cost homomorphisms to locally semicomplete digraphs and quasi-transitive
                  digraphs},
  journal      = {Australas. {J} Comb.},
  volume       = {46},
  pages        = {217--232},
  year         = {2010},
  url          = {http://ajc.maths.uq.edu.au/pdf/46/ajc\_v46\_p217.pdf},
  bibsource    = {dblp computer science bibliography, https://dblp.org}
}

@article{DBLP:journals/jgt/Bang-JensenCM17,
  author       = {J{\o}rgen Bang{-}Jensen and
                  Tilde My Christiansen and
                  Alessandro Maddaloni},
  title        = {Disjoint Paths in Decomposable Digraphs},
  journal      = {J. Graph Theory},
  volume       = {85},
  number       = {2},
  pages        = {545--567},
  year         = {2017},
  url          = {https://doi.org/10.1002/jgt.22090},
  doi          = {10.1002/JGT.22090},
  bibsource    = {dblp computer science bibliography, https://dblp.org}
}

@article{DBLP:journals/jgt/Bang-JensenGY03,
  author       = {J{\o}rgen Bang{-}Jensen and
                  Gregory Z. Gutin and
                  Anders Yeo},
  title        = {Steiner type problems for digraphs that are locally semicomplete or
                  extended semicomplete},
  journal      = {J. Graph Theory},
  volume       = {44},
  number       = {3},
  pages        = {193--207},
  year         = {2003},
  url          = {https://doi.org/10.1002/jgt.10140},
  doi          = {10.1002/JGT.10140},
  bibsource    = {dblp computer science bibliography, https://dblp.org}
}

@article{DBLP:journals/jgt/GuoV96,
  author       = {Yubao Guo and
                  Lutz Volkmann},
  title        = {Locally semicomplete digraphs that are complementary \emph{m}-pancyclic},
  journal      = {J. Graph Theory},
  volume       = {21},
  number       = {2},
  pages        = {121--136},
  year         = {1996},
  url          = {https://doi.org/10.1002/(SICI)1097-0118(199602)21:2\<121::AID-JGT2\>3.0.CO;2-T},
  doi          = {10.1002/(SICI)1097-0118(199602)21:2\<121::AID-JGT2\>3.0.CO;2-T},
  bibsource    = {dblp computer science bibliography, https://dblp.org}
}

@article{DBLP:journals/jgt/Bang-JensenGV96,
  author       = {J{\o}rgen Bang{-}Jensen and
                  Yubao Guo and
                  Lutz Volkmann},
  title        = {Weakly Hamiltonian-connected locally semicomplete digraphs},
  journal      = {J. Graph Theory},
  volume       = {21},
  number       = {2},
  pages        = {163--172},
  year         = {1996},
  url          = {https://doi.org/10.1002/(SICI)1097-0118(199602)21:2\<163::AID-JGT5\>3.0.CO;2-P},
  doi          = {10.1002/(SICI)1097-0118(199602)21:2\<163::AID-JGT5\>3.0.CO;2-P},
  bibsource    = {dblp computer science bibliography, https://dblp.org}
}

@article{DBLP:journals/jgt/Guo96,
  author       = {Yubao Guo},
  title        = {Strongly Hamiltonian-connected locally semicomplete digraphs},
  journal      = {J. Graph Theory},
  volume       = {22},
  number       = {1},
  pages        = {65--73},
  year         = {1996},
  url          = {https://doi.org/10.1002/(SICI)1097-0118(199605)22:1\<65::AID-JGT9\>3.0.CO;2-J},
  doi          = {10.1002/(SICI)1097-0118(199605)22:1\<65::AID-JGT9\>3.0.CO;2-J},
  bibsource    = {dblp computer science bibliography, https://dblp.org}
}

@article{DBLP:journals/jgt/GuoV94,
  author       = {Yubao Guo and
                  Lutz Volkmann},
  title        = {Connectivity properties of locally semicomplete digraphs},
  journal      = {J. Graph Theory},
  volume       = {18},
  number       = {3},
  pages        = {269--280},
  year         = {1994},
  url          = {https://doi.org/10.1002/jgt.3190180306},
  doi          = {10.1002/JGT.3190180306},
  bibsource    = {dblp computer science bibliography, https://dblp.org}
}

@article{DBLP:journals/jgt/Bang-JensenH14,
  author       = {J{\o}rgen Bang{-}Jensen and
                  Jing Huang},
  title        = {Arc-Disjoint In- and Out-Branchings With the Same Root in Locally
                  Semicomplete Digraphs},
  journal      = {J. Graph Theory},
  volume       = {77},
  number       = {4},
  pages        = {278--298},
  year         = {2014},
  url          = {https://doi.org/10.1002/jgt.21786},
  doi          = {10.1002/JGT.21786},
  bibsource    = {dblp computer science bibliography, https://dblp.org}
}

@article{DBLP:journals/jct/Bang-JensenH12,
  author       = {J{\o}rgen Bang{-}Jensen and
                  Jing Huang},
  title        = {Decomposing locally semicomplete digraphs into strong spanning subdigraphs},
  journal      = {J. Comb. Theory {B}},
  volume       = {102},
  number       = {3},
  pages        = {701--714},
  year         = {2012},
  url          = {https://doi.org/10.1016/j.jctb.2011.09.001},
  doi          = {10.1016/J.JCTB.2011.09.001},
  bibsource    = {dblp computer science bibliography, https://dblp.org}
}

@article{DBLP:journals/jgt/Bang-JensenM14,
  author       = {J{\o}rgen Bang{-}Jensen and
                  Alessandro Maddaloni},
  title        = {Arc-Disjoint Paths in Decomposable Digraphs},
  journal      = {J. Graph Theory},
  volume       = {77},
  number       = {2},
  pages        = {89--110},
  year         = {2014},
  url          = {https://doi.org/10.1002/jgt.21775},
  doi          = {10.1002/JGT.21775},
  bibsource    = {dblp computer science bibliography, https://dblp.org}
}

@article{DBLP:journals/jgt/Bang-JensenH95,
  author       = {J{\o}rgen Bang{-}Jensen and
                  Jing Huang},
  title        = {Quasi-transitive digraphs},
  journal      = {J. Graph Theory},
  volume       = {20},
  number       = {2},
  pages        = {141--161},
  year         = {1995},
  url          = {https://doi.org/10.1002/jgt.3190200205},
  doi          = {10.1002/JGT.3190200205},
  bibsource    = {dblp computer science bibliography, https://dblp.org}
}

@inproceedings{DBLP:conf/icalp/GhorbaniM26,
  author       = {Ebrahim Ghorbani and
                  Matthias Mnich},
  editor       = {Sayan Bhattacharya and
                  Danupon Nanongkai and
                  Michael Benedikt and
                  Gabriele Puppis},
  title        = {A 9/4-Approximation for Directed Feedback Vertex Sets in Quasi-Transitive
                  Digraphs},
  booktitle    = {53rd International Colloquium on Automata, Languages, and Programming,
                  {ICALP} 2026, Royal Holloway, University of London, Egham, United
                  Kingdom, July 7-10, 2026},
  series       = {LIPIcs},
  volume       = {374},
  pages        = {96:1--96:16},
  publisher    = {Schloss Dagstuhl - Leibniz-Zentrum f{\"{u}}r Informatik},
  year         = {2026},
  url          = {https://doi.org/10.4230/LIPIcs.ICALP.2026.96},
  doi          = {10.4230/LIPICS.ICALP.2026.96},
  bibsource    = {dblp computer science bibliography, https://dblp.org}
}

@article{DBLP:journals/talg/LokshtanovMMPPS21,
  author = {Daniel Lokshtanov and Pranabendu Misra and
            Joydeep Mukherjee and Fahad Panolan and
            Geevarghese Philip and Saket Saurabh},
  title = {2-Approximating Feedback Vertex Set in Tournaments},
  journal = {ACM Transactions on Algorithms},
  volume = {17},
  number = {2},
  pages = {11:1--11:14},
  year = {2021},
  doi = {10.1145/3446969},
  url = {https://doi.org/10.1145/3446969}
}

@article{DBLP:journals/jgt/Skrien82,
  author       = {Dale Skrien},
  title        = {A relationship between triangulated graphs, comparability graphs,
                  proper interval graphs, proper circular-arc graphs, and nested interval
                  graphs},
  journal      = {J. Graph Theory},
  volume       = {6},
  number       = {3},
  pages        = {309--316},
  year         = {1982},
  url          = {https://doi.org/10.1002/jgt.3190060307},
  doi          = {10.1002/JGT.3190060307},
  bibsource    = {dblp computer science bibliography, https://dblp.org}
}

\appendix
\section*{Appendix}
\label{appendix}



\section{Amplifying the Semicomplete Routine}
\label{sc:amplify}

For an integer $k\ge1$, the routine
\textsc{AmplifiedSEM-FVS}$(D,w,k)$ runs $\mathcal S(D,w)$
independently $k$ times and returns a minimum-weight output.

\begin{lemma}
\label{lem:tfvs-amplify}
On an $n$-vertex semicomplete instance $(D,w)$ and an integer $k\ge1$,
the amplified routine always returns an FVS $S$ and
satisfies
$\Pr[w(S)\le2\operatorname{OPT}(D,w)]\ge1-2^{-k}$.
Its running time is $\mathcal O(kn^{17})$.
\end{lemma}

\begin{proof}
Every repetition returns an FVS, so the lightest output is always
feasible. If at least one repetition satisfy its factor-$2$ guarantee, then
the chosen output is no heavier and satisfies the same guarantee.
By Theorem~\ref{thm:semfvs-2approx}, each independent repetition
fails with probability at most $1/2$, so all $k$ fail with probability
at most $2^{-k}$. Finally, the $k$ calls each take
$\mathcal O(n^{17})$ time, giving the stated running time.
\end{proof}

\end{document}